\documentclass{article} 
\usepackage{amssymb}
\usepackage{amsfonts}
\usepackage{amsmath}
\usepackage{amsthm}
\usepackage{graphicx}
\usepackage{setspace}

\newcommand{\compconj}[1]{%
  \overline{#1}%
}

\newtheorem{theorem}{Theorem}[section]

\newtheorem{lemma}[theorem]{Lemma}
\newtheorem{proposition}[theorem]{Proposition}

\newtheorem{assumption}[theorem]{Assumption}
\newtheorem{definition}[theorem]{Definition}
\providecommand{\keywords}[1]{\textbf{Keywords---} #1}

\title{Nonlocal Diffusions and The Quantum Black-Scholes Equation: Modelling the Market Fear Factor.}
\author{Will HICKS$^{\ast}$\thanks{$^\ast$ 42 Cranes Park Avenue, Surbiton, KT5 8BP, United Kingdom\newline Email: whicks7940@googlemail.com}}
\begin{document}
\maketitle
\begin{abstract}
In this paper, we establish a link between quantum stochastic processes, and nonlocal diffusions. We demonstrate how the non-commutative Black-Scholes equation of Accardi \& Boukas (cf \cite{AccBoukas}) can be written in integral form. This enables the application of the Monte-Carlo methods adapted to McKean stochastic differential equations (cf \cite{McKean}) for the simulation of solutions. We show how unitary transformations can be applied to classical Black-Scholes systems to introduce novel quantum effects. These have a simple economic interpretation as a market `fear factor', whereby recent market turbulence causes an increase in volatility going forward, that is not linked to either the local volatility function or an additional stochastic variable. Lastly, we extend this system to 2 variables, and consider Quantum models for bid-offer spread dynamics.
\end{abstract}
\section{Introduction}
The link between the classical Black-Scholes equation and quantum mechanics and the application of quantum formalism to Mathematical Finance has been investigated by several authors. For example: \cite{AccBoukas}-\cite{Choustova}, \cite{Haven}-\cite{Hidalgo}, \cite{Maslov}, and \cite{McCloud_1}-\cite{Segal}. In particular, the approach of modelling derivative prices using self-adjoint operators on a Hilbert space was suggested by Segal \& Segal in \cite{Segal}. In this paper the authors noted that, in the real world, the market operates with imperfect information and that different observables, such as underlying price and option delta, are usually not simultaneously observable. This fact makes the non-commutative extension of the Black-Scholes framework a natural step. The authors point out that this approach addresses some of the limitations of the classical Black-Scholes model, such as the underestimation of the probability of extreme events - so called ``fat tails". In this sense, non-commutative Quantum models present an alternative means to capture complex market dynamics, without the addition of new stochastic variables.\newline
In \cite{AccBoukas}, Accardi \& Boukas derive a general form for the Quantum Black-Scholes equation based on the Hudson-Parthasarathy calculus (cf \cite{Hudson-Parthasarathy}) and show that a commutative unitary time development operator acting on the market state, leads to a classical Black-Scholes system. Further they give the quantum stochastic differential equation governing the time development operator, and demonstrate how unitary transformations can lead to non-commutativity. An example of a non-commutative Quantum Black-Scholes partial differential equation is derived, although the authors work in an abstract setting and do not discuss specific unitary transformations and Hilbert space representations of financial markets.\newline
Therefore, one objective of this work is to use the Accardi-Boukas framework to look at how different unitary transformations can be used to transform the classical Black-Scholes equation, and to understand how quantum effects become apparent. We then go on to explore an example application of the approach in the modelling of bid-offer spread dynamics. The final objective of the current work is to identify suitable Monte-Carlo methods, which can be used for the simulation of solutions.\newline
In section 2, we give an overview of the Accardi-Boukas derivation of the general form for the Quantum Black-Scholes equation, from \cite{AccBoukas}. With the objective of looking at ``near classical'' Black-Scholes worlds, we then derive specific forms for the resulting partial differential equations that result from small translations, and rotations. This in turn involves the extension of the Accardi-Boukas equation to systems with more than one underlying variable. We go on to discuss how this approach can be applied to the modelling of bid-offer dynamics.\newline
In section 3, we show how this can be linked to the nonlocal diffusion processes discussed by Luczka, H{\"a}nggi and Gadomski in \cite{Luczka}. Here the impact of the diffusion differential operator is spread out through the convolution with a ``blurring" function. The Kramers-Moyal expansion of the nonlocal Fokker-Planck equations allows us to derive the moments of the blurring function for the ``near classical'' quantum system.\newline
This approach allows a natural route to the visualisation of the quantum effects on the system using McKean SDEs (cf \cite{McKean}). The Monte-Carlo methods, developed by Guyon, and Henry-Labord\`ere in \cite{Guyon}, can then be adapted to the simulation of solutions. This is discussed in section 4, where we present numerical results and show how, by introducing small transformations to the system, the stochastic process now reacts to a market downturn by returning higher volatility. This effect is observed even where there is a single static Black-Scholes type volatility.
\section{Quantum Black-Scholes equation}
In this section we follow the notation given, by Accardi \& Boukas, in \cite{AccBoukas}. The current market is represented by a vector in a Hilbert space: $\mathbb{H}$, which contains all relevant information about the state of the market at an instant in time. The tradeable price for a security is represented by an self-adjoint operator on $\mathbb{H}$: $X$, and the the spectrum of $X$ represents possible prices.\newline
\break
Let $L^2[\mathbb{R}^+;\mathbb{H}]$ represent functions from the positive real axis (time) to the Hilbert space $\mathbb{H}$. Then the random behaviour of tradeable securities can be modelled using the tensor product of $\mathbb{H}$ with the bosonic Fock space: $\mathbb{H}\otimes\Gamma (L^2[\mathbb{R}^+;\mathbb{H}])$. We term this the ``market space". The operator that returns the current price becomes $X\otimes \mathbb{I}$, where $\mathbb{I}$ represents the identity operator. The time development of $X\otimes \mathbb{I}$ into the future is modelled by:\newline
\break
$j_t(X)=U_t^*X\otimes\mathbb{I}U_t$\newline
\break
$\mathbb{H}$ carries the initial state of the market and $U_t$ acts by introducing random fluctuations that fill up the empty states in: $\Gamma (L^2[\mathbb{R}^+;\mathbb{H}])$. The functional form for $U_t$ is derived by Hudson \& Parthasarathy in \cite{Hudson-Parthasarathy}, and is given by:\newline
\break
$dU_t=-\Bigg(\bigg(iH+\frac{1}{2}L^*L\bigg)dt+L^*SdA_t-LdA^{\dagger}_t+\bigg(1-S\bigg)d\Lambda_t\Bigg)U_t$\newline
\break
$dA^{\dagger}_t, dA_t, d\Lambda_t$ represent the standard creation, annihilation, and Poisson operators of quantum stochastic calculus. $H, S$ and $L$ also operate on the market space, with $S$ unitary, and $H$ self-adjoint. The multiplication rules of the Hudson-Parthasarathy calculus are given below (cf \cite{Hudson-Parthasarathy}):\newline
\break
\begin{tabular}{p{1cm}|p{1cm}p{1cm}p{1cm}p{1cm}}
-&$dA^{\dagger}_t$&$d\Lambda_t$&$dA_t$&$dt$\\
\hline
$dA^{\dagger}_t$&0&0&0&0\\
$d\Lambda_t$&$dA^{\dagger}_t$&$d\Lambda_t$&0&0\\
$dA_t$&$dt$&$dA_t$&0&0\\
$dt$&0&0&0&0
\end{tabular}\newline
\break
The first thing to note is that, for $S\neq 1$, there is a non-zero Poisson term and the time development operator is non-commutative.\newline
\break
The next thing to note is that, where $S=1$, the Poisson term disappears. The model can be written using the Ito calculus in place of the more general Hudson-Parthasarathy framework. The Wiener process $dW_t$ can be modelled using: $dA_t+dA^{\dagger}_t$.\newline
\break
Let $V_T=j_T(X-K)^+$, represent the option price process as at final expiry $T$, and $K$ the operator given by multiplying by the strike. Further, for $V_t=j_t(X-K)^+$ the following expansion is assumed:\newline
\break
$V_t=F(t,x)=\sum_{n,k} a_{n,k}(t-t_0)^n(x-x_0)^k$\newline
\break
The Hudson-Parthasarathy multiplication rules can be applied to this expansion to give a quantum stochastic differential equation for $V_t$, that corresponds to the usual Ito expansion used in the derivation of the classical Black-Scholes. By assuming one can construct a hedge portfolio by holding the underlying and a risk free numeraire asset, Accardi \& Boukas are able to derive the general form the Quantum Black Scholes equation using the assumption that any portfolio must be self financing. Proposition 1, from \cite{AccBoukas} gives the full Quantum Black-Scholes equation:\newline
\break
\begin{equation}\label{QBS}
a_{1,0}(t,j_t(X))+a_{0,1}(t,j_t(X))j_t(\theta)+\sum_{k=2}^{\infty} a_{0,k}(t,j_t(X))j_t(\alpha\lambda^{k-2}\alpha^{\dagger})=a_tj_t(\theta)+V_tr-a_tj_t(X)r
\end{equation}
Here, $a_t$ represents the holding in the underlying asset and is given by the boundary conditions:\newline
\break
$\sum_{k=1}^{\infty} a_{0,k}(t,j_t(X))j_t(\lambda^{k-1}\alpha^{\dagger})=a_tj_t(\alpha^{\dagger})$\newline
$\sum_{k=1}^{\infty} a_{0,k}(t,j_t(X))j_t(\alpha\lambda^{k-1})=a_tj_t(\alpha)$\newline
$\sum_{k=1}^{\infty} a_{0,k}(t,j_t(X))j_t(\lambda^{k})=a_tj_t(\lambda)$\newline
\break
Further, $\theta, \alpha$ and $\lambda$ are given by:\newline
\break
$\alpha=[L^*,X]S$, $\lambda=S^*XS-X$, $\theta=i[H,X]-\frac{1}{2}\{L^*LX+XL^*L+2L^*XL\}$. In this case the boundary conditions arise because when the Poisson term: $d\Lambda$ is non-zero, unlike Ito calculus where expansion terms with order above 2 can be ignored, higher order terms still contain non-vanishing contribution.
\subsection{Translation}
The natural Hilbert space for an equity price (say the FTSE price) is: $\mathbb{H}=L^2[\mathbb{R}]$. In this case, the only unitary transactions we can use are the translations:\newline
\break
$T_\epsilon:f(x)\rightarrow f(x-\varepsilon)$\newline
\break
Here we have, for a translation invariant Lebesgue measure $\mu$:\newline
\break
$\langle T_\varepsilon f|T_\varepsilon g\rangle = \int_{\mathbb{R}} \compconj{f(x-\varepsilon)}g(x-\varepsilon)d\mu=\int_{\mathbb{R}} \compconj{f(x)}g(x)d\mu=\langle f|g\rangle$\newline
\break
So S is unitary in this case. Therefore, translating by $\varepsilon$ we get:\newline
\break
$\lambda=T_{-\varepsilon}XT_\varepsilon f(x) - Xf(x)=T_{-\varepsilon}xf(x-\varepsilon )-xf(x)=(x+\varepsilon )f(x)-xf(x)=\varepsilon f(x)$\newline
\break
So we have $\lambda=\varepsilon$, and it is clear the example given in \cite{AccBoukas} relates to a translation by $\varepsilon=1$. Following the key steps from \cite{AccBoukas} Proposition 3, and inserting this back into equation \ref{QBS}, we get the following Quantum Black-Scholes partial differential equation for this system:\begin{lemma}
Let $u(t,x)$ represent the price at time t, of a derivative contract in the system described above under small translation $\varepsilon$, and with interest rate $r$. Then the quantum Black-Scholes equation becomes:
\begin{equation} \label{eq:1}
\frac{\partial u(t,x)}{\partial t}=rx\frac{\partial u(t,x)}{\partial x}-u(t,x)r+\sum_{k=2}^{\infty}\frac{\varepsilon^{k-2}}{k!}\frac{\partial^k u(t,x)}{\partial x^k}g(x)
\end{equation}
\end{lemma}
\begin{proof}
The proof follows the same steps Accardi \& Boukas outline in \cite{AccBoukas} proposition 3, with small modifications.
\end{proof}
For $\varepsilon=0$, the last term drops out, and the equation reverts to the classical Black-Scholes. We investigate the impact of non-zero $\varepsilon$ in section 4.
\subsection{Rotation}
For the one dimensional market space: $L^2[\mathbb{R}]$, the Lebesgue invariant translations, are the only unitary transformations available. However, the true current state of the financial market contains a much richer variety of information than just a single price, and by increasing the dimensionality of the Market space accordingly we introduce a wider variety of unitary transformations, that can introduce non-commutativity. For example, let $x$ represent the FTSE mid-price, and $\epsilon$ half of the bid-offer spread so that $(x+\epsilon)$ represents the best offer-price and $(x-\epsilon)$ the best bid-price. Now the market is represented by the Hilbert space: $\mathbb{H}=L^2[\mathbb{R}^2]$, and we can apply rotations, in addition to translations.\newline
\break
We make the simplifying assumption that market participants can trade the mid-price: $x$ (for example during the end of day auction process) and that the market has sufficient liquidity to enable participants to alternatively act as market makers (receiving bid-offer spread) or as hedgers (crossing bid-offer spread) and therefore trade the bid-offer spread: $\epsilon$. Therefore we make the following assumption:
\begin{assumption}\label{ass_1}
For any derivative payout $V(x_T,\epsilon_T)$, we can construct a hedged portfolio, and can proceed with the derivation of the Quantum Black Scholes equation following the basic methodology from \cite{AccBoukas}.
\end{assumption}
We now have separate creation, annihilation and Poisson operators, for $x$ and $\epsilon$; $dA_x$, $dA_{\epsilon}$ etc. These can be combined using the multiplication table (\cite{Hudson-Parthasarathy}, Theorem 4.5), by making the assumption that the bid-offer is uncorrelated with the equity price. This corresponds to assumption 
\ref{ass_2}:
\begin{assumption}\label{ass_2}
$dA_xd\Lambda_{\epsilon}=dA_{\epsilon}d\Lambda_x=d\Lambda_xd\Lambda_{\epsilon}=d\Lambda_{\epsilon}d\Lambda_x=dA_xdA^{\dagger}_{\epsilon}=dA_{\epsilon}dA^{\dagger}_x=d\Lambda_xdA^{\dagger}_{\epsilon}=d\Lambda_{\epsilon}dA^{\dagger}_x=0$.
\end{assumption}
Lastly, we make the assumption that we can expand the derivative payout as before:
\begin{assumption}\label{ass_3}
$V_t=F(t,x,\epsilon)=\sum_{n,k,l} a_{n,l,k}(t-t_0)^n(x-x_0)^k(\epsilon-\epsilon_0)^l$
\end{assumption}
We can now derive the relevant Quantum Black-Scholes equation:
\begin{proposition}
Let $\mathbb{H}=L^2[\mathbb{R}^2]$, and let $X\otimes 1$ and $\epsilon\otimes 1$ operate on the market space: $\mathbb{H}\otimes\Gamma (L^2[\mathbb{R}^+;\mathbb{H}])$, to return the mid-price, and bid-offer spread for a tradeable security respectively. Further, let the notation from \cite{AccBoukas}, and the above assumptions apply.\newline
\break
Then the Quantum Black-Scholes equation in this case is given by:\newline
\begin{equation}\label{QBS2}
\begin{split}
a_{1,0,0}(t,j_t(X),j_t(\epsilon))+a_{0,1,0}(t,j_t(X),j_t(\epsilon))j_t(\theta_x)+a_{0,0,1}(t,j_t(X),j_t(\epsilon))j_t(\theta_{\epsilon})\\
+\sum_{k=2}^{\infty} a_{0,k,0}(t,j_t(X),j_t(\epsilon))j_t(\alpha_x\lambda_x^{k-2}\alpha_x^{\dagger})+\sum_{l=2}^{\infty} a_{0,0,l}(t,j_t(X),j_t(\epsilon))j_t(\alpha_{\epsilon}\lambda_{\epsilon}^{l-2}\alpha_{\epsilon}^{\dagger})\\
=a_{x,t}j_t(\theta_x)+a_{\epsilon,t}j_t(\theta_{\epsilon})+V_tr-a_{x,t}j_t(X)r-a_{\epsilon,t}j_t(\epsilon)r
\end{split}
\end{equation}
Where for $j_t(X)$:
\begin{equation}\label{bc_x}
\begin{split}
\sum_{k=1}^{\infty} a_{0,k,0}(t,j_t(X),j_t(\epsilon))j_t(\lambda_x^{k-1}\alpha_x^{\dagger})=a_{x,t}j_t(\alpha_x^{\dagger})\\
\sum_{k=1}^{\infty} a_{0,k,0}(t,j_t(X),j_t(\epsilon))j_t(\alpha_x\lambda_x^{k-1})=a_{x,t}j_t(\alpha_x)\\
\sum_{k=1}^{\infty} a_{0,k,0}(t,j_t(X),j_t(\epsilon))j_t(\lambda_x^{k})=a_{x,t}j_t(\lambda_x)\\
\end{split}
\end{equation}
and for $j_t(\epsilon)$:\newline
\begin{equation}\label{bc_e}
\begin{split}
\sum_{l=1}^{\infty} a_{0,0,l}(t,j_t(X),j_t(\epsilon))j_t(\lambda_{\epsilon}^{l-1}\alpha_{\epsilon}^{\dagger})=a_{\epsilon,t}j_t(\alpha_{\epsilon}^{\dagger})\\
\sum_{l=1}^{\infty} a_{0,0,l}(t,j_t(X),j_t(\epsilon))j_t(\alpha_{\epsilon}\lambda_{\epsilon}^{l-1})=a_{\epsilon,t}j_t(\alpha_{\epsilon})\\
\sum_{l=1}^{\infty} a_{0,0,l}(t,j_t(X),j_t(\epsilon))j_t(\lambda_{\epsilon}^{l})=a_{\epsilon,t}j_t(\lambda_{\epsilon})
\end{split}
\end{equation}
\end{proposition}
\begin{proof}
First, the equations for time-development operators for $X\otimes 1$, and $\epsilon\otimes 1$ become:\newline
\break
$dU_{x,t}=-\Bigg(\bigg(iH+\frac{1}{2}L_x^*L_x\bigg)dt+L_x^*SdA_x-L_xdA_x^{\dagger}+\bigg(1-S\bigg)d\Lambda_x$\newline
$dU_{\epsilon,t}=-\Bigg(\bigg(iH+\frac{1}{2}L_{\epsilon}^*L_{\epsilon}\bigg)dt+L_{\epsilon}^*SdA_{\epsilon}-L_{\epsilon}dA_{\epsilon}^{\dagger}+\bigg(1-S\bigg)d\Lambda_{\epsilon}$\newline
\break
Then, applying the Hudson-Parthasarathy multiplication rules to the expansion given in assumption \ref{ass_3} gives:
\begin{equation}\label{expansion}
\begin{split}
dV_t=\bigg(a_{1,0,0}(t,j_t(x),j_t(\epsilon))+a_{0,1,0}(t,j_t(x),j_t(\epsilon))j_t(\theta_x)+a_{0,0,1}(t,j_t(x),j_t(\epsilon))j_t(\theta_{\epsilon})\\
+\sum_{k=2}^{\infty} a_{0,k,0}(t,j_t(X),j_t(\epsilon))j_t(\alpha_x\lambda_x^{k-2}\alpha_x^{\dagger})+\sum_{l=2}^{\infty} a_{0,0,l}(t,j_t(X),j_t(\epsilon))j_t(\alpha_{\epsilon}\lambda_{\epsilon}^{l-2}\alpha_{\epsilon}^{\dagger})\bigg)dt\\
+\bigg(a_{0,1,0}(t,j_t(X),j_t(\epsilon))j_t(\alpha_x)+\sum_{k=2}^{\infty} a_{0,k,0}(t,j_t(X),j_t(\epsilon))j_t(\alpha_x\lambda_x^{k-1})\bigg)dA_x\\
+\bigg(a_{0,0,1}(t,j_t(X),j_t(\epsilon))j_t(\alpha_{\epsilon})+\sum_{l=2}^{\infty} a_{0,0,l}(t,j_t(X),j_t(\epsilon))j_t(\alpha_{\epsilon}\lambda_{\epsilon}^{k-1})\bigg)dA_{\epsilon}\\
+\bigg(a_{0,1,0}(t,j_t(X),j_t(\epsilon))j_t(\alpha_x^{\dagger})+\sum_{k=2}^{\infty} a_{0,k,0}(t,j_t(X),j_t(\epsilon))j_t(\lambda_x^{k-1}\alpha_x^{\dagger})\bigg)dA_x^{\dagger}\\
+\bigg(a_{0,0,1}(t,j_t(X),j_t(\epsilon))j_t(\alpha_{\epsilon}^{\dagger})+\sum_{l=2}^{\infty} a_{0,0,l}(t,j_t(X),j_t(\epsilon))j_t(\lambda_{\epsilon}^{k-1}\alpha_{\epsilon}^{\dagger})\bigg)dA_{\epsilon}^{\dagger}
\end{split}
\end{equation}
Where $\theta_x, \theta_{\epsilon}$ are given by:\newline
\break
$\theta_x=i[H,X]-\frac{1}{2}\bigg(L_x^*L_xX+XL_x^*L_x-2L_x^*XL_x\bigg)$\newline
$\theta_{\epsilon}=i[H,\epsilon]-\frac{1}{2}\bigg(L_{\epsilon}^*L_{\epsilon}\epsilon+\epsilon L_{\epsilon}^*L_{\epsilon}-2L_{\epsilon}^*\epsilon L_{\epsilon}\bigg)$\newline
\break
$\alpha_x, \alpha_{\epsilon}$ are given by:\newline
\break
$\alpha_x=[L_x^*,X]S$\newline
$\alpha_{\epsilon}=[L_{\epsilon}^*,\epsilon]S$\newline
\break
and finally $\lambda_x, \lambda_{\epsilon}$ are given by:\newline
\break
$\lambda_x=S^*XS-X$\newline
$\lambda_{\epsilon}=S^*\epsilon S-\epsilon$\newline
\break
By assumption \ref{ass_1} we can form a hedge portfolio which we now use:\newline
\break
$V_t=a_{x,t}j_t(X)+a_{\epsilon,t}j_t(\epsilon)+b_t\beta$, for risk free numeraire asset $\beta$.\newline
\break
$dV_t=a_{x,t}dj_t(X)+a_{\epsilon,t}dj_t(\epsilon)+b_t\beta rdt$\newline
\break
Applying the unitary time development operators for $\epsilon$ and $x$ we have:\newline
\break
\begin{equation}\label{hedges}
\begin{split}
dV_t=a_{x,t}\big(j_t(\alpha_x^{\dagger})dA_x^{\dagger}+j_t(\lambda_x)d\Lambda_x+j_t(\alpha_x)dA_x\big)\\
+a_{\epsilon,t}\big(j_t(\alpha_{\epsilon}^{\dagger})dA_{\epsilon}^{\dagger}+j_t(\lambda_{\epsilon})d\Lambda_{\epsilon}+j_t(\alpha_{\epsilon})dA_{\epsilon}\big)\\
+\big(j_t(\theta_x)+(V_t-a_{x,t}j_t(X)-a_{\epsilon,t}j_t(\epsilon))r\big)dt
\end{split}
\end{equation}
\break
Equating the risky terms between equations (\ref{expansion}), and (\ref{hedges}) leads to the boundary conditions, (\ref{bc_x}) and (\ref{bc_e}) on $a_{x,t}$ and $a_{\epsilon,t}$. Similarly, equating the $dt$ terms, leads to the Quantum Black-Scholes equation for this system: equation (\ref{QBS2}).
\end{proof}
Now, let $f\big(x,\epsilon\big)$ represent a vector in $\mathbb{H}$, and apply a rotation matrix:\newline
$S=\begin{bmatrix} cos(\phi) && -sin(\phi) \\ sin(\phi) && cos(\phi) \end{bmatrix}$\newline
\break
We have:\newline
\break
$Sf\big(x,\epsilon\big)=f\big(cos(\phi)x-sin(\phi)\epsilon,cos(\phi)\epsilon+sin(\phi)x\big)$\newline
\break
$XSf=xf\big(cos(\phi)x-sin(\phi)\epsilon,cos(\phi)\epsilon+sin(\phi)x\big)$\newline
\break
$S^*XSf=\big(cos(\phi)x+sin(\phi)\epsilon\big)f(x,\epsilon)$\newline
\break
So, we end up with:\newline
\break
$\lambda_x=\bigg(\big(cos(\phi)-1\big)x+sin(\phi)\epsilon\bigg)$, $\lambda_{\epsilon}=\bigg(\big(cos(\phi)-1\big)\epsilon-sin(\phi)x\bigg)$.\newline
\break
Finally, inserting this back into equation (\ref{QBS2}), we get the Black-Scholes equation for the system (following notation from \cite{AccBoukas}):
\begin{proposition}
Let $u(t,x,\epsilon)$ represent the price at time t, of a derivative contract in the system described above under rotation $\phi$, and with interest rate $r$. Then the quantum Black-Scholes equation becomes: 
\begin{equation}\label{eq:2}
\begin{split}
\frac{\partial u(t,x,\epsilon)}{\partial t}=rx\frac{\partial u(t,x,\epsilon)}{\partial x}+r\epsilon\frac{\partial u(t,x,\epsilon)}{\partial \epsilon}-u(t,x,\epsilon)r\\
+\sum_{k=2}^{\infty}\frac{((cos(\phi)-1)x+sin(\phi)\epsilon)^{k-2}}{k!}\frac{\partial^k u(t,x,\epsilon)}{\partial x^k}g_1(x,\epsilon)\\
+\sum_{l=2}^{\infty}\frac{((cos(\phi)-1)\epsilon-sin(\phi)x)^{l-2}}{l!}\frac{\partial^l u(t,x,\epsilon)}{\partial \epsilon^l}g_2(x,\epsilon)
\end{split}
\end{equation}
\end{proposition}
\begin{proof}
We assume that the operators $L_x, L_x^*,L_{\epsilon}, L_{\epsilon}^*$ involve multiplication by a polynomial in $x, \epsilon$, and therefore commute with $\lambda_x, \lambda_{\epsilon}$. Therefore, from the boundary conditions we have:\newline
\break
$\sum_{k=1}^{\infty} a_{0,k,0}(t,j_t(X),j_t(\epsilon))j_t(\lambda_x^{k-1})=a_{x,t}$\newline
\break
$\sum_{l=1}^{\infty} a_{0,0,l}(t,j_t(X),j_t(\epsilon))j_t(\lambda_{\epsilon}^{l-1})=a_{\epsilon,t}$\newline
\break
Inserting this into \ref{QBS2} gives:
\begin{equation}
\begin{split}
a_{1,0,0}(t,j_t(X),j_t(\epsilon))+a_{0,1,0}(t,j_t(X),j_t(\epsilon))j_t(X)r+a_{0,0,1}(t,j_t(X),j_t(\epsilon))j_t(\epsilon)r\\
+\sum_{k=2}^{\infty} a_{0,k,0}(t,j_t(X),j_t(\epsilon))j_t(\lambda_x^{k-2}(\alpha_x\alpha_x^*-\lambda_x(\theta_x-xr)))\\
+\sum_{l=2}^{\infty} a_{0,0,l}(t,j_t(X),j_t(\epsilon))j_t(\lambda_{\epsilon}^{l-2}(\alpha_{\epsilon}\alpha_{\epsilon}^*-\lambda_{\epsilon}(\theta_{\epsilon}-\epsilon r)))\\
=V_tr
\end{split}
\end{equation}
Now writing $g_1(x,\epsilon)=j_t(\alpha_x\alpha_x^*-\lambda_x(\theta_x-xr))$, $g_2(x,\epsilon)=j_t(\alpha_{\epsilon}\alpha_{\epsilon}^*-\lambda_{\epsilon}(\theta_{\epsilon}-\epsilon r))$, and $a_{0,k,0}(t,j_t(X),j_t(\epsilon))=\frac{1}{k!}\frac{\partial^k u}{\partial x^k}, a_{0,0,l}(t,j_t(X),j_t(\epsilon))=\frac{1}{l!}\frac{\partial^l u}{\partial {\epsilon}^l}$, we have the result given.
\end{proof}
For small rotations, we have $cos(\phi)=1-\frac{\varepsilon^2}{2}+o(\varepsilon^2)$, and $sin(\phi)=\varepsilon + o(\varepsilon^2)$. Inserting this into equation (\ref{eq:2}), we have a new partial differential equation, where the coefficient of the $k$th partial derivative, for $k\geq 3$, with respect to $x, \epsilon$, is correct to $o(\varepsilon^{2(k-2)})$. This form for small rotations is more amenable to the methods we apply in section 3.
\begin{equation}\label{eq:3}
\begin{split}
\frac{\partial u(t,x,\epsilon)}{\partial t}=rx\frac{\partial u(t,x,\epsilon)}{\partial x}+r\epsilon\frac{\partial u(t,x,\epsilon)}{\partial \epsilon}-u(t,x,\epsilon)r\\
+\sum_{k=2}^{\infty}\frac{(\varepsilon\epsilon-(\varepsilon^2/2)x)^{k-2}}{k!}\frac{\partial^k u(t,x,\epsilon)}{\partial x^k}g_1(x,\epsilon)\\
+\sum_{l=2}^{\infty}\frac{(-\varepsilon x-(\varepsilon^2/2)\epsilon)^{l-2}}{l!}\frac{\partial^l u(t,x,\epsilon)}{\partial \epsilon^l}g_2(x,\epsilon)
\end{split}
\end{equation}
As is the case for equation (\ref{eq:1}), this reduces to the classical Black-Scholes for 2 uncorrelated random variables (in this case price: $x$, and bid-offer spread: $\epsilon$) when $\varepsilon=0$.\newline
\break
For the classical case, the addition of the bid-offer spread is in some ways unnecessary when using the model for derivative pricing. For derivative contracts depending on the close price, one can usually hedge daily at the closing price during the end of day auction process. For many trading desks this may be sufficient in practice, and terms involving the bid-offer spread will drop out of the model. In the quantum case, examination of equations (\ref{eq:2}) and (\ref{eq:3}) shows that we expect interference between the bid-offer spread dynamics and the price dynamics. For small rotations, these equations are singular PDEs, and we expect the behaviour in most regions to approximate classical behaviour. However, when the higher derivative terms are larger, quantum interference may be significant. We discuss this more in sections 3 and 4.
\section{Nonlocal Diffusions}
In this section, we derive the Fokker-Planck equations associated to the Quantum Black-Scholes equations: (\ref{eq:1}), and (\ref{eq:3}). We show how these can be written in integral form, by using the Kramers-Moyal expansion (see for example \cite{Frank}). This enables us to link the Quantum Black-Scholes models of the previous section to nonlocal diffusions (see for example the paper by Luczka, H{\"a}nggi and Gadomski: \cite{Luczka}). We assume zero interest rates in this section to help clarify the notation without changing the key dynamics. The integral form for the Fokker-Planck equations is given by:
\begin{equation}\label{FokkerPlanck1}
\begin{split}
\frac{\partial p(t,x,\epsilon)}{\partial t}=\frac{1}{2}\frac{\partial^2}{\partial x^2}\bigg(\int_{-\infty}^{\infty}\int_{-\infty}^{\infty}\big(H(y_x,y_{\epsilon}|x,\epsilon)g_1(x,\epsilon)p(x-y_x,\epsilon-y_{\epsilon},t)\big)dy_xdy_{\epsilon}\bigg)\\
+\frac{1}{2}\frac{\partial^2}{\partial {\epsilon}^2}\bigg(\int_{-\infty}^{\infty}\int_{-\infty}^{\infty}\big(H(y_x,y_{\epsilon}|x,\epsilon)g_2(x,\epsilon)p(x-y_x,\epsilon-y_{\epsilon},t)\big)dy_xdy_{\epsilon}\bigg)
\end{split}
\end{equation}
The function $H(y_x,y_{\epsilon}|x,\epsilon)$ has the effect of "blurring" the impact of the diffusion operator. In the case that $H(y_x,y_{\epsilon}|x,\epsilon)$ is a Dirac delta function, the diffusion operator is localised as usual, and the associated Fokker-Planck equation reduces to the standard Kolmogorov forward equation associated with the classical Black-Scholes.\newline
\break
We start with the following general form for equations (\ref{eq:1}) and (\ref{eq:3}):
\begin{equation}\label{eq:4}
\frac{\partial u(t,x,\epsilon)}{\partial t}=g_1(x,\epsilon)\sum_{k=2}^{\infty}\frac{f_1(x,\epsilon,\varepsilon)^{k-2}}{k!}\frac{\partial^k u(t,x,\epsilon)}{\partial x^k}+g_2(x,\epsilon)\sum_{l=2}^{\infty}\frac{f_2(x,\epsilon,\varepsilon)^{l-2}}{l!}\frac{\partial^l u(t,x,\epsilon)}{\partial \epsilon^l}
\end{equation}
\begin{proposition}
The Fokker-Planck equation associated to equation (\ref{eq:4}), with $r=0$ is given by:
\begin{equation}\label{FokkerPlanck2}
\begin{split}
\frac{\partial p(t,x,\epsilon)}{\partial t}=\sum_{k=2}^{\infty}\frac{(-1)^k}{k!}\frac{\partial^k \big(g_1(x,\epsilon)f_1(x,\epsilon,\varepsilon)^{k-2}p(t,x,\epsilon)\big)}{\partial x^k}\\
+\sum_{l=2}^{\infty}\frac{(-1)^l}{l!}\frac{\partial^l \big(g_2(x,\epsilon)f_2(x,\epsilon,\varepsilon)^{l-2}p(t,x,\epsilon)\big)}{\partial \epsilon^l}
\end{split}
\end{equation}
\end{proposition}
\begin{proof}
For a derivative payout $h(x,\epsilon)$, with zero interest rates, we have the following price in risk neutral measure $Q$:\newline
\break
$u(x_t,\epsilon_t, t)=E^{Q}\big[h(x_T,\epsilon_T)\big]=\int_{\mathbb{R}^2} h(y_x,y_{\epsilon})p(y_x,y_{\epsilon}|x,\epsilon,t)dy_xdy_{\epsilon}$\newline
\break
Where $p(y_x,y_{\epsilon}|x,\epsilon,t)$ represents the risk neutral probability density for the variables observed at time $T$, conditional on the values at time $t$. $h(x,\epsilon)$ represents a derivative payout at $T$. We then write the right hand integral as:\newline
\break
$\int_{\mathbb{R}^2} g(y_x,y_{\epsilon})p(y_x,y_{\epsilon}|x,\epsilon,t)dy_xdy_{\epsilon}=\int_0^t\int_{\mathbb{R}^2} Lh(y_x,y_{\epsilon})p(y_x,y_{\epsilon}|x,\epsilon,s)dy_xdy_{\epsilon}ds$\newline
\break
Where $L$ represents the operator:\newline
\break
$Lh(x,\epsilon)=\bigg(g_1(x,\epsilon)\sum_{k=2}^{\infty}\frac{f_1(x,\epsilon,\varepsilon)^{k-2}}{k!}\frac{\partial^k}{\partial x^k}+g_2(x,\epsilon)\sum_{l=2}^{\infty}\frac{f_2(x,\epsilon,\varepsilon)^{l-2}}{l!}\frac{\partial^l}{\partial \epsilon^l}\bigg)h(x,\epsilon)$\newline
\break
The Fokker-Planck equation, is given by the adjoint operator $L^*$. Therefore, since:\newline
\break
$\int_0^t\int_{\mathbb{R}^2} Lh(y_x,y_{\epsilon})p(y_x,y_{\epsilon}|x,\epsilon,s)dy_xdy_{\epsilon}ds=\int_0^t\int_{\mathbb{R}^2} h(y_x,y_{\epsilon})L^*p(y_x,y_{\epsilon}|x,\epsilon,s)dy_xdy_{\epsilon}ds$\newline
\break
If we truncate equation (\ref{eq:4}) at a certain order for the derivative: $N$, the result follows by integrating by parts $N$ times. Proceeding with higher and higher $N$, we can match the derivative terms of any arbitary order, and the result follows.
\end{proof}
The objective now, is to write equation (\ref{FokkerPlanck2}) in the form of (\ref{FokkerPlanck1}). To do this we can follow a Moment Matching algorithm. We use the following expansion:\newline
\break
$g(x,\epsilon)p(x-y_x,\epsilon-y_{\epsilon},t)=\sum_{i,j=0}^{\infty} \frac{(-1)^{(i+j)}}{(i+j)!}y_x^iy_{\epsilon}^j\frac{d^{i+j}(g(x,\epsilon)p(x,\epsilon))}{dx^id{\epsilon}^j}$\newline
\break
Inserting this into equation (\ref{FokkerPlanck1}) gives:
\begin{equation}\label{eq:5}
\begin{split}
\frac{\partial p(t,x,\epsilon)}{\partial t}=\frac{1}{2}\frac{\partial^2}{\partial x^2}\bigg(\sum_{i,j=0}^{\infty}\frac{(-1)^{(i+j)}}{(i+j)!}\frac{\partial^{i+j}(g_1(x,\epsilon)p(x,\epsilon))}{\partial x^i\partial\epsilon^j}\int_{-\infty}^{\infty}\int_{-\infty}^{\infty}H(y_x,y_{\epsilon}|x,\epsilon)y_x^iy_{\epsilon}^jdy_xdy_{\epsilon}\bigg)\\
+\frac{1}{2}\frac{\partial^2}{\partial {\epsilon}^2}\bigg(\sum_{i,j=0}^{\infty}\frac{(-1)^{(i+j)}}{(i+j)!}\frac{\partial^{i+j}(g_2(x,\epsilon)p(x,\epsilon))}{\partial x^i\partial\epsilon^j}\int_{-\infty}^{\infty}\int_{-\infty}^{\infty}H(y_x,y_{\epsilon}|x,\epsilon)y_x^iy_{\epsilon}^jdy_xdy_{\epsilon}\bigg)
\end{split}
\end{equation}
Now by equating the coefficients of the derivatives with respect to $x$ and $\epsilon$, between equations (\ref{eq:5}) and (\ref{FokkerPlanck2}) one can calculate the moments of the ``blurring" function $H(y_x,y_{\epsilon}|x,\epsilon)$. For the translation case, $g_2(x,\epsilon)=0$, and the probability density is a function of $x$ only.
\subsection{Moment Matching: Translation Case}
In the translation case, of section 2.1, since the coefficients of each differential term in equation (\ref{eq:1}) is a constant multiplied by $g(x)$, the moments of the ``blurring'' function $H(y)$ will not depend of $x$. Equation (\ref{eq:5}) becomes:
\begin{equation}\label{eq:6}
\frac{\partial p(t,x)}{\partial t}=\frac{1}{2}\bigg(\sum_{j=0}^{\infty}\frac{(-1)^{(j)}}{j!}\frac{d^{(j+2)}(g(x)p(x))}{dx^{(j+2)}}\int_{-\infty}^{\infty}H(y)y^jdy\bigg)
\end{equation}
Similarly, the Fokker-Planck associated with equation (\ref{eq:1}), with $r=0$, is given by:
\begin{equation}\label{eq:7}
\frac{\partial p(t,x)}{\partial t}=\sum_{k=2}^{\infty}\frac{(-1)^k\varepsilon^{k-2}}{k!}\frac{\partial^k (g(x)p(t,x))}{\partial x^k}
\end{equation}
Now the moments of the ``blurring'' function can be matched by equating directly equations (\ref{eq:6}) and (\ref{eq:7}):
\begin{proposition}\label{trans_mom}
Let $H_i$ represent the $i^{th}$ moment of $H(y)$, for the Fokker-Planck equation (\ref{FokkerPlanck2}), relating to the translation case described in section 2.1. Then, $H_i$ is given by:\newline
\break
$H_i=\frac{2(-\varepsilon)^i}{(i+1)(i+2)}$
\end{proposition}
\begin{proof}
$H_i$ follows (for $i\geq 0$) by equating the coefficients for: $\frac{\partial^{(i+2)}}{\partial x^{(i+2)}}$, between equations (\ref{eq:6}) and (\ref{eq:7}).
\end{proof}
We find that, in this case, $H(y)$ is a normalised function that tends to a Dirac function as $\varepsilon$ tends to zero, and for $\varepsilon=0$ we end up with classical 2nd order Fokker-Planck equation. This is discussed further in section 4.
\subsection{Moment Matching: Rotation Case}
In the rotation case of section 2.2, the coefficients of each differential term in equation (\ref{FokkerPlanck2}) are functions of $x$ and $\epsilon$. Therefore, we require the moments for the ``blurring'' function also to be functions of $x$, and $\epsilon$: $H(y_x,y_{\epsilon}|e,\epsilon)$. Once we have calculated the coefficients for the differential terms, we can use these to form an inhomogeneous 2nd order differential equation for the moments of $H(y_x,y_{\epsilon}|e,\epsilon)$.\newline
\break
In this case, from equation (\ref{FokkerPlanck2}) we have: $f_1(x,\epsilon)=\varepsilon\epsilon-(\varepsilon^2/2)x$, and $f_2(x,\epsilon)=-\varepsilon x-(\varepsilon^2/2)\epsilon$. Therefore, the Fokker-Planck equation associated with equation (\ref{eq:3}), with $r=0$, is given by:
\begin{equation}\label{eq:8}
\begin{split}
\frac{\partial p(t,x,\epsilon)}{\partial t}=\sum_{k=2}^{\infty}\frac{1}{k!}\frac{\partial^k \bigg(\big((\varepsilon^2/2)x-\varepsilon\epsilon\big)^{k-2}g_1(x,\epsilon)p(t,x,\epsilon)\bigg)}{\partial x^k}\\
+\sum_{l=2}^{\infty}\frac{1}{l!}\frac{\partial^l \bigg(\big(\varepsilon x+(\varepsilon^2/2)\epsilon\big)^{l-2}g_2(x,\epsilon)p(t,x,\epsilon)\bigg)}{\partial \epsilon^l}
\end{split}
\end{equation}
The moments of the ``blurring'' function will now follow by equating coefficients for the differential terms between equations (\ref{eq:5}), and (\ref{eq:8}).
\begin{proposition}\label{rot_mom}
Where the moments of the ``blurring'' function: $H(y_x,y_{\epsilon}|x,\epsilon)$ are given by:\newline
\break
$a^i_{x}=\int_{-\infty}^{\infty}\int_{-\infty}^{\infty}H(y_x,y_{\epsilon}|x,\epsilon)y_x^idy_xdy_{\epsilon}$\newline
$a^j_{\epsilon}=\int_{-\infty}^{\infty}\int_{-\infty}^{\infty}H(y_x,y_{\epsilon}|x,\epsilon)y_{\epsilon}^jdy_xdy_{\epsilon}$\newline
\break
and $a^0, a^1_x, a^1_{\epsilon}$ are assumed to be:\newline
\break
$a^0=1, a^1_x, a^1_{\epsilon}=0$\newline
\break
Then for the higher moments we have, for $n\geq 2$:
\begin{equation}
\frac{(-1)^na^{n-2}_x+2n\frac{\partial a^{n-1}_x}{\partial x}+n(n-1)\frac{\partial^2 a^{n}_x}{\partial x^2}}{n!}=\frac{((\varepsilon^2/2)x-\varepsilon\epsilon)^{n-2}}{n(n-1)(1-(\varepsilon^2/2))^{(n-1)}}
\end{equation}
\begin{equation}
\frac{(-1)^na^{n-2}_{\epsilon}+2n\frac{\partial a^{n-1}_{\epsilon}}{\partial \epsilon}+n(n-1)\frac{\partial^2 a^{n}_{\epsilon}}{\partial {\epsilon}^2}}{n!}=\frac{((\varepsilon^2/2)\epsilon+\varepsilon x)^{n-2}}{n(n-1)(1-(\varepsilon^2/2))^{(n-1)}}
\end{equation}
\end{proposition}
\begin{proof}
We first calculate the coefficients for $\frac{\partial^n (g_1(x,\epsilon)p(x,\epsilon))}{\partial x^n}$ from equation (\ref{eq:8}).\newline
\break
The 2nd order coefficient is given by:\newline
\break
$\sum_{i\geq 2} \frac{(i-2)!(\varepsilon^2/2)^{i-2}\binom{i}{2}}{i!}=\frac{1}{2}\sum_{i\geq 0} (\varepsilon^2/2)^i=\frac{1}{2(1-(\varepsilon^2/2))}$\newline
\break
Similarly, the 3rd order coefficient is given by:\newline
\break
$\sum_{i\geq 3} \frac{(i-2)!(\varepsilon^2/2)^{(i-2)}\binom{i}{3}((\varepsilon^2/2)x-\varepsilon\epsilon)}{i!}=\frac{((\varepsilon^2/2)x-\varepsilon\epsilon)}{3!}\sum_{i\geq 0} (i+1)(\varepsilon^2/2)^i=\frac{((\varepsilon^2/2)x-\varepsilon\epsilon)}{3!(1-(\varepsilon^2/2))^2}$\newline
\break
In general, the nth order coefficient is given by:\newline
\break
$\sum_{i\geq n} \frac{(i-2)!(\varepsilon^2/2)^{(i-2)}\binom{i}{n}((\varepsilon^2/2)x-\varepsilon\epsilon)^{n-2}}{i!(n-2)!}$\newline
\break
$=\frac{((\varepsilon^2/2)x-\varepsilon\epsilon)^{(n-2)}}{n!}\sum_{i\geq 0} (i+1)(i+2)...(i+n-2)(\varepsilon^2/2)^i$\newline
\break
The final summation can be calculated by differentiating $(n-2)$ times, the infinite sum $1/(1-v)$, where $v=(\varepsilon^2/2)$.\newline
\break
Therefore, the coefficient for $n\geq 2$ is given by:
\begin{equation}\label{x_coeff}
\frac{((\varepsilon^2/2)x-\varepsilon\epsilon)^{n-2}}{n(n-1)(1-(\varepsilon^2/2))^{(n-1)}}\frac{\partial^n (g_1(x,\epsilon)p(x,\epsilon))}{\partial x^n}
\end{equation}
Following similar logic for $\epsilon$ we have the coefficient:
\begin{equation}\label{epsilon_coeff}
\frac{((\varepsilon^2/2)\epsilon+\varepsilon x)^{n-2}}{n(n-1)(1-(\varepsilon^2/2))^{(n-1)}}\frac{\partial^n (g_2(x,\epsilon)p(x,\epsilon))}{\partial \epsilon^n}
\end{equation}
These coefficients can now be used to calculate a 2nd order inhomogeneous differential equation for the moments of $H(y_x,y_{\epsilon}|x,\epsilon)$. We start by expanding the $\partial^2/\partial x^2$, and $\partial^2/\partial\epsilon^2$ in equation (\ref{eq:5}).\newline
\break
Since, we assume from section 2.2, that $x,\epsilon$ are uncorrelated, equation (\ref{eq:5}) can be written:
\begin{equation}\label{eq:9}
\begin{split}
\frac{\partial p(t,x,\epsilon)}{\partial t}=
\frac{1}{2}\sum_{i=0}^{\infty}\frac{(-1)^{(i)}}{i!}\Bigg(\frac{\partial^{i}(g_1(x,\epsilon)p(x,\epsilon))}{\partial x^i}\frac{\partial^2 a^i_{x}}{\partial x^2}+\frac{\partial^{i+2}(g_1(x,\epsilon)p(x,\epsilon))}{\partial x^{i+2}}a^i_{x}\\+2\frac{\partial^{i+1}(g_1(x,\epsilon)p(x,\epsilon))}{\partial x^{i+1}}\frac{\partial a^i_{x}}{\partial x}\Bigg)\\
+\frac{1}{2}\sum_{j=0}^{\infty}\frac{(-1)^{(j)}}{j!}\Bigg(\frac{\partial^{j}(g_2(x,\epsilon)p(x,\epsilon))}{\partial \epsilon^j}\frac{\partial^2 a^j_{\epsilon}}{\partial \epsilon^2}+\frac{\partial^{j+2}(g_1(x,\epsilon)p(x,\epsilon))}{\partial \epsilon^j}a^j_{\epsilon}\\+2\frac{\partial^{j+1}(g_1(x,\epsilon)p(x,\epsilon))}{\partial \epsilon^{j+1}}\frac{\partial a^j_{\epsilon}}{\partial \epsilon}\Bigg)\\
\end{split}
\end{equation}
The coefficients for $\frac{\partial^n (g_1(x,\epsilon)p(x,\epsilon))}{\partial x^n}$ from equation (\ref{eq:9}) are now given by: $\frac{\partial^2 a^0}{\partial x^2}(g_1(x,\epsilon)p(x,\epsilon))$ for $n=0$, $(\frac{\partial^2 a^1_x}{\partial x^2}+2\frac{\partial a^0}{\partial x})\frac{\partial (g_1(x,\epsilon)p(x,\epsilon))}{\partial x}$ for $n=1$, and:
\begin{equation}\label{eq:10}
\frac{(-1)^na^{n-2}_x+2n\frac{\partial a^{n-1}_x}{\partial x}+n(n-1)\frac{\partial^2 a^{n}_x}{\partial x^2}}{n!}\frac{\partial^n (g_1(x,\epsilon)p(x,\epsilon))}{\partial x^n}
\end{equation}
\ for $n\geq 2$. Similarly, for $\epsilon$ we have:
\begin{equation}\label{eq:11}
\frac{(-1)^na^{n-2}_{\epsilon}+2n\frac{\partial a^{n-1}_{\epsilon}}{\partial \epsilon}+n(n-1)\frac{\partial^2 a^{n}_{\epsilon}}{\partial {\epsilon}^2}}{n!}\frac{\partial^n (g_2(x,\epsilon)p(x,\epsilon))}{\partial \epsilon^n}
\end{equation}
We now make the assumption that $H$ is a normalised probability distribution with expectation zero for $x$ and $\epsilon$. Ie, $\frac{\partial a_0}{\partial x}=0$, $a^1_x=0$, and $a^1_{\epsilon}=0$. These assumptions ensure the coefficients with $n=0,1$ equate to zero on both sides of equation (\ref{eq:8}). The proposition follows by equating equations (\ref{x_coeff})/(\ref{eq:10}) and (\ref{epsilon_coeff})/(\ref{eq:11}).
\end{proof}
\section{Monte-Carlo Methods \& Numerical Simulations}
In this section, we give a brief overview of McKean stochastic differential equations, before introducing how the particle method, discussed in the book by Guyon \& Henry-Labord\`ere: \cite{Guyon}, can be used in their simulation. We then go on to present numerical results from the bid-offer model discussed above, placing particular emphasis on understanding how quantum effects become apparent through small transformations applied to a classical Black-Scholes system.
\subsection{McKean Stochastic Differential Equations}
McKean nonlinear stochastic differential equations were introduced in \cite{McKean}, and refer to SDEs, where the drift \& volatility coefficients depend on the underlying probability law for the stochastic process. Following notation from \cite{Guyon} we have:\newline
\break
$dX_t=b(t,X_t,\mathbb{P}_t)dt+\sigma(t,X_t,\mathbb{P}_t)dW_t$\newline
\break
These are then related to the nonlinear Fokker Planck equation:
\begin{equation}
\frac{\partial p}{\partial t}=\frac{1}{2}\sum_{i,j}\frac{\partial^2 (\sigma_i(t,x,\mathbb{P}_t)\sigma_j(t,x,\mathbb{P}_t)p(t,x))}{\partial x_i\partial x_j}-\sum_{i}\frac{\partial (b^i(t,x,\mathbb{P}_t))}{\partial x_i}
\end{equation}
In this case, we can write equation (\ref{FokkerPlanck1}) in this form. We have for $r=0$, $b^1(t,x,\epsilon,\mathbb{P}_t)=b^2(t,x,\epsilon,\mathbb{P}_t)=0$ and $\sigma_1(t,x,\epsilon,\mathbb{P}_t)=\sqrt{g_1(x,\epsilon)\mathbb{E}^p\bigg[\frac{H(x-y_x,\epsilon-y_{\epsilon}|x,\epsilon)}{p(x,\epsilon,t)}\bigg]}$, $\sigma_2(t,x,\epsilon,\mathbb{P}_t)=\sqrt{g_2(x,\epsilon)\mathbb{E}^p\bigg[\frac{H(x-y_x,\epsilon-y_{\epsilon}|x,\epsilon)}{p(x,\epsilon,t)}\bigg]}$.\newline
\break
Therefore, we can simulate the solution to equation (\ref{FokkerPlanck1}) by first calculating the function $H(x-y_x,\epsilon-y_{\epsilon})$ using a moment matching algorithm, and then simulating the following McKean SDE, with uncorrelated Wiener processes $dW^1, dW^2$:
\begin{equation}\label{McKean}
\begin{split}
dx=\sqrt{\frac{g_1(x,\epsilon)}{p(x,\epsilon,t)}\mathbb{E}^{p(y)}\big[H(x-y_x,\epsilon-y_{\epsilon}|x,\epsilon)\big]}dW^1\\
d\epsilon=\sqrt{\frac{g_2(x,\epsilon)}{p(x,\epsilon,t)}\mathbb{E}^{p(y)}\big[H(x-y_{\epsilon},\epsilon-y_{\epsilon}|x,\epsilon)\big]}dW^2\\
\end{split}
\end{equation}
The simulation of the above SDE relies on the {\it particle method} outlined in Guyon \& Henry-Labord\`ere's book {\it Nonlinear Option Pricing} chapters 10, 11 (cf: \cite {Guyon}).\newline
\break
Each path $(x^i,\epsilon^i)$ now interacts with the other paths: $(x^j,\epsilon^j), j\neq i$ during the simulation process, and the convergence of the method relies on the so called {\it propagation of the chaos} property. This states:
\begin{definition}
For all functions $\phi(x,\epsilon,t)\in C_0(\mathbb{R}^2)$:
\begin{equation}
\frac{1}{N}\sum_{j=1}^{N} \phi(x^j,\epsilon^j)\xrightarrow{N\rightarrow\infty} \int_{\mathbb{R}^2} \phi(x,\epsilon,t)p(x,\epsilon,t)dxd\epsilon
\end{equation}
\end{definition}
In our case, the SDE (\ref{McKean}), is a McKean-Vlasov process, and we have from Guyon, Henry-Labord\`ere (cf: \cite{Guyon} Theorem 10.3), and originally Sznitman (cf: \cite{Sznitman}), that the propagation of the chaos property holds.
\subsection{Particle Method}
The first step is to discretize the SDE: (\ref{McKean}), as follows:
\begin{equation}\label{discrete_SDE}
\begin{split}
dx^i=\bigg(\sum_{j=1}^{N} H(x^j-x^i,\epsilon^j-\epsilon^i)\frac{P(x^j,\epsilon^j)}{P(x^i,\epsilon^i)}g_1(x^i,\epsilon^i)\bigg)^{0.5}dW^{1,i}\\
d\epsilon^i=\bigg(\sum_{j=1}^{N} H(x^j-x^i,\epsilon^j-\epsilon^i)\frac{P(x^j,\epsilon^j)}{P(x^i,\epsilon^i)}g_2(x^i,\epsilon^i)\bigg)^{0.5}dW^{2,i}
\end{split}
\end{equation}
Where $P(x^j,\epsilon^j)$ represents a suitably discretized probability function. The algorithm then proceeds as follows:
\begin{enumerate}
\item Solve for the moments of the ``blurring'' function $H(x-y_x,\epsilon-y_{\epsilon}|x,\epsilon)$ using propositions \ref{trans_mom}, and \ref{rot_mom}.
\item Choose a parameterised distribution to approximate $H(x-y_x,\epsilon-y_{\epsilon}|x,\epsilon)$, and fit the parameters using the calculated moments. For example, approximate $H(x-y_x,\epsilon-y_{\epsilon}|x,\epsilon)$ as a univariate/bivariate normal distribution.
\item Simulate the 1st time step, $t_1$, using the value of $H(0,0|x_0,\epsilon_0)$, for starting positions $x_0,\epsilon_0$.
\item After each simulation, allocate the simulated paths into discrete probability buckets: $P(x^j,\epsilon^j)$, for paths $j=1$ to $N$.
\item Proceed from the $t_{k-1}$ to $t_k$ timestep, using (\ref{discrete_SDE}), the value of $H(x-y_x,\epsilon-y_{\epsilon}|x,\epsilon)$, and the discrete buckets at $t_{k-1}$.
\item Iterate steps 4 \& 5 until the final maturity: $t_F$.
\end{enumerate}
\subsection{Modelling the Market Fear Factor}
We can see from (\ref{discrete_SDE}), that small translations, will lead to a variance scaling factor:\newline
\break
$\sum_{j=1}^{N} H(x^j-x^i,\epsilon^j-\epsilon^i)\frac{P(x^j,\epsilon^j)}{P(x^i,\epsilon^i)}$\newline
\break
This will have the impact of reducing the volatility of those paths which lie in the middle of the ``bell curve'', owing to the negative curvature of the probability law at these points - probability mass is spread by the ``blurring'' function to lower probability points.\newline
\break
Similarly, at the extremes of the probability density curve where the curvature is positive, probability mass is spread to areas with net higher probability. In essence the market memory of a recent extreme event, will lead to a higher market volatility at the next time step.\newline
\break
This effect differs from the negative skew observed in local volatility models (for example the work by Dupire: cf \cite{Dupire}), and from stochastic volatility models (for example Heston: \cite{Heston}), in the sense that the increase in volatility is linked to recent random moves in the tail of the probability distribution, rather than to the level of the stochastic volatility or a static function of the price, and time.\newline
\break
To highlight the difference, in the process given by equation (\ref{discrete_SDE}), one could allow for periodic rebalancing of the process. For example, one could replace the unconditional probability, with the probability conditional on the previous step. In this way, the level of the volatility would depend purely on a ``memory'' of recent price history, rather than on the absolute level of the market price, or an additional random variable. The market responds to large moves with a heightened fear factor. The study of modelling such processes with rebalancing, will involve advanced techniques for calculating the conditional probabilities, and we defer detailed study to a future work.\newline
\break
\subsection{Numerical Results}
In this section, we simulate the one-factor process described in section 2.1, and 3.1. In this case, we approximate $H(y)$ using a normal distribution using the moments from proposition 3.2: $N(\frac{\varepsilon}{3},\frac{\varepsilon^2}{18})$.\newline
\break
The non-zero 1st moment, will lead to an upside/downside bias to the ``market fear factor'' effect. Essentially, by introducing a translation in the negative $x$ direction, one introduces downside `fear' into the model.\newline
\break
Figures 1 \& 2 below, illustrate the results from a 2 step Monte-Carlo process, with $g(x)=0.01x^2$, starting value: $x_0=1$, 100K Monte-Carlo paths, and 500 discrete probability buckets. The scatter plot shows the magnitude of the proportional return on the 1st time-step on the horizontal axis, and the second time-step on the vertical axis:\newline
\break
\includegraphics[scale=0.5]{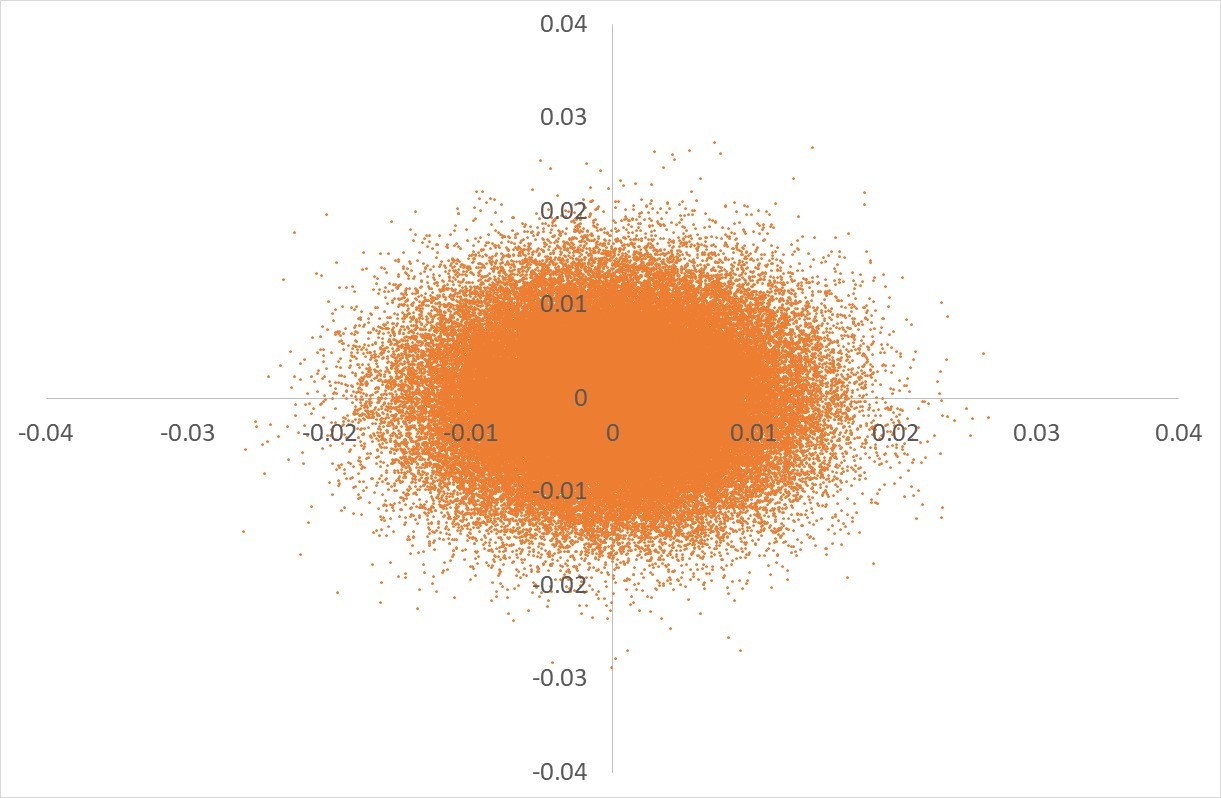}\newline
Figure 1: $\varepsilon=0$, horizontal axis represents the proportional return for the first time-step, vertical axis represents the second second time-step.\newline
\break
\includegraphics[scale=0.5]{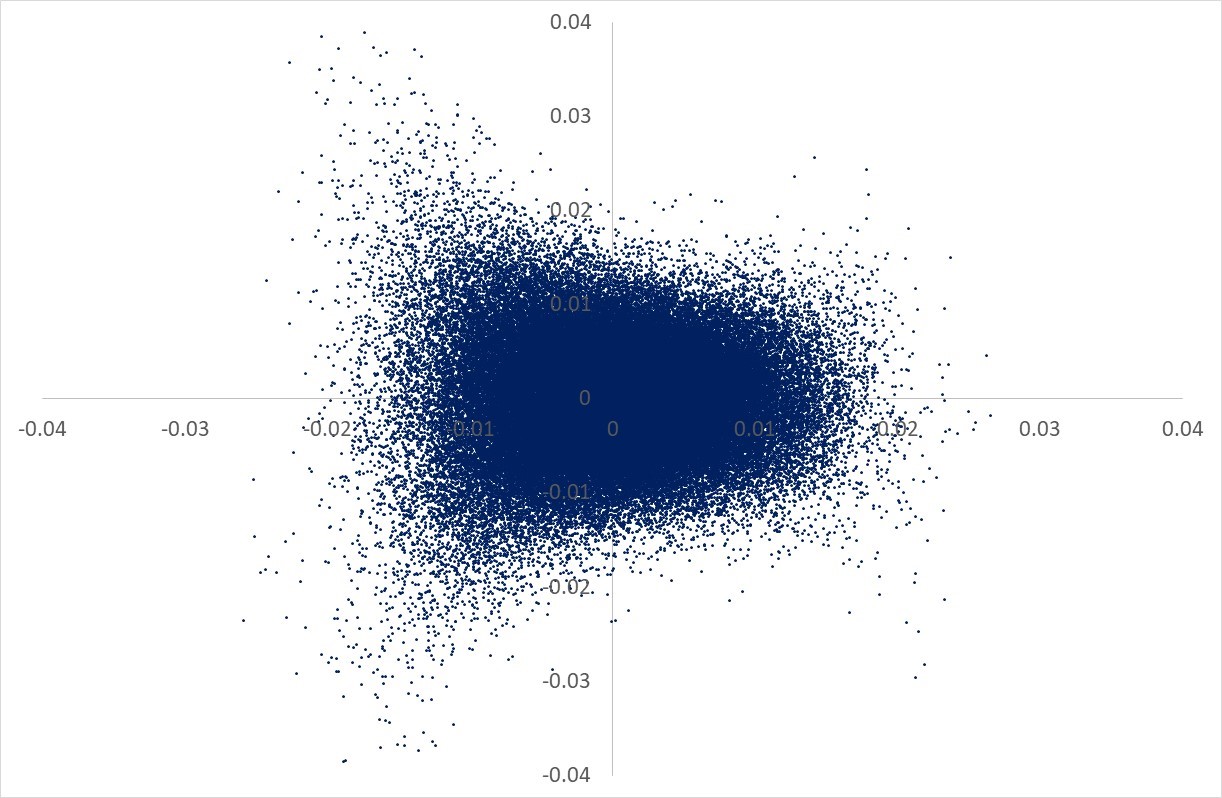}\newline
Figure 2: The results for $\varepsilon=0.02$, horizontal axis represents the proportional return for the first time-step, vertical axis represents the second time-step.\newline
\break
Figure 1 shows the results for $\varepsilon=0$. This is a classical Black-Scholes system, and there is no correlation between the magnitude \& direction of the 1st and 2nd time-steps.\newline
\break
Figure 2 shows the proportional returns for $\varepsilon=0.02$ (in blue), overlaid on top of the $\varepsilon=0$ results (in orange). The volatility of the second step is reduced on those paths where the first time-step has been small. There is a slight increased second step volatility for those paths with large positive first steps, and significant second step volatility for those paths with a large negative first step. In effect, the drop in market prices has introduced ``fear'' into these paths.\newline
\break
The final chart shows the probability distributions for the natural logarithm of the simulated value after 50 one day time-steps. The non-zero translation results in a natural skewness in the distribution.\newline
\break
\includegraphics[scale=0.5]{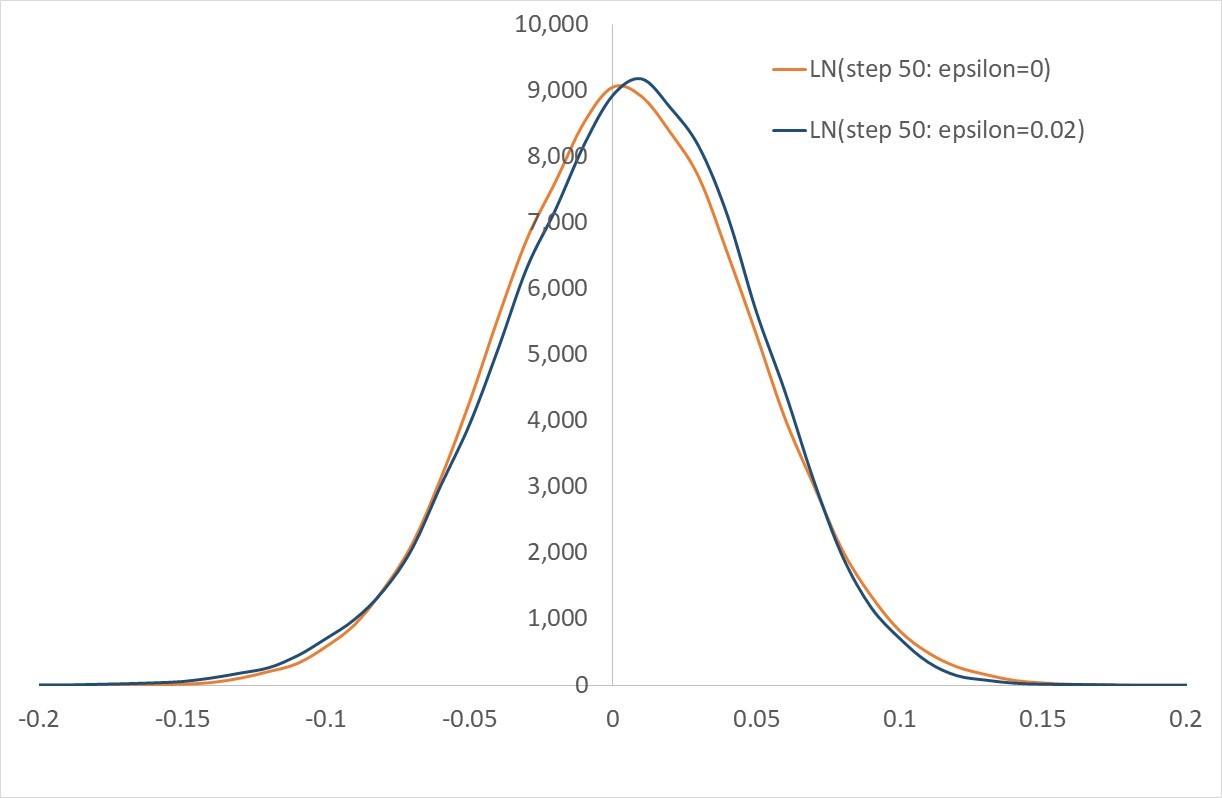}\newline
Figure 3: Distribution for the natural log of the final price after 50 one day time-steps. 100K Monte-Carlo paths, and 500 discrete probability buckets.
\section{Conclusions}
In this paper, we demonstrate how unitary transformations can be used to model novel quantum effects in the Quantum Black-Scholes system of Accardi \& Boukas (cf \cite{AccBoukas}).\newline
We show how these quantum stochastic processes can also be modelled using nonlocal diffusions, and simulated using the particle method outlined by Guyon \& Henry-Labord\`ere in \cite{Guyon}.\newline
By introducing a bid-offer spread parameter, and extending the Accardi-Boukas framework to 2 variables, we show how rotations, in addition to translations, can be applied. Thus, a richer representation of the information contained in the current market leads to a wider variety of unitary transformations that can be used.\newline
In section 4, using a Monte-Carlo simulation, we illustrate how introducing a translation to the one dimensional model leads to a skewed distribution, whereby recent market down moves leads to increased volatility going forward. In effect, the market retains memory of recent significant moves.\newline
In \cite{Dupire}, Dupire shows how to calibrate a local volatility to the current vanilla option smile. This enables a Monte-Carlo simulation that is fully consistent with current market option prices. Carrying out the same analysis, using the new Quantum Fokker-Planck equations, is another important next step to consider as a future development of the current work.

\end{document}